\documentclass[journal, 11pt, onecolumn]{IEEEtran}
\IEEEoverridecommandlockouts

\usepackage{cite,soul}
\usepackage{amsmath,amssymb,amsfonts,amsthm}
\usepackage{algorithmic}
\usepackage{graphicx}
\usepackage{textcomp,verbatim}
\usepackage{xcolor}
\usepackage{setspace} 
\usepackage{blkarray}

\newtheorem{theorem}{Theorem}

\newtheorem{corollary}{Corollary}
\newtheorem{lemma}{Lemma}

\newtheorem{definition}{Definition}

\theoremstyle{remark}
\newtheorem{example}{Example}
\newtheorem{remark}{Remark}

\begin{document}

\title{Function-Correcting Codes for Locally \\ Bounded Functions}

\author{\IEEEauthorblockN{ Charul Rajput\IEEEauthorrefmark{1}, B. Sundar Rajan\IEEEauthorrefmark{2}, Ragnar Freij-Hollanti\IEEEauthorrefmark{1}, Camilla Hollanti\IEEEauthorrefmark{1}}

\IEEEauthorblockA{\IEEEauthorrefmark{1}{\normalsize Department of Mathematics and System Analysis, Aalto University, Finland
    \\\{charul.rajput, ragnar.freij, camilla.hollanti\}@aalto.fi}}

    \IEEEauthorblockA{\IEEEauthorrefmark{2}{\normalsize Department of Electrical Communication Engineering, Indian Institute of Science, Bengaluru, India
    \\\ bsrajan@iisc.ac.in}}}

\maketitle

\begin{abstract}
In this paper, we introduce a class of functions that assume only a limited number $\lambda$ of values within a given Hamming $\rho$-ball and call them locally $(\rho, \lambda)$-bounded functions. We develop function-correcting codes (FCCs) for a subclass of these functions and propose an upper bound on the redundancy of FCCs. The bound is based on the minimum length of an error-correcting code with a given number of codewords and a minimum distance. Furthermore, we provide a sufficient optimality condition for FCCs when $\lambda = 4$. We also demonstrate that any function can be represented as a locally $(\rho, \lambda)$-bounded function, illustrating this with a representation of Hamming weight distribution functions. Furthermore, we present another construction of function-correcting codes for Hamming weight distribution functions. 
\end{abstract}


\section{Introduction}
Function-correcting codes (FCCs) are a class of codes introduced by Lenz et al. in \cite{LBWY2023}. These codes are designed to protect the evaluation of a specific function of message vectors during transmission over noisy channels. Unlike traditional error-correcting codes (ECCs), which aim to protect the entire message vector against errors, FCCs focus on preserving particular attributes or functions of the message. This targeted approach is particularly useful in numerous practical scenarios \cite{MW1967,BK1981,AC1981,OR2001,BNZ2009,KW2017,SSGAGD2019}, where the focus is primarily on a specific attribute of the message. It is also more efficient than protecting the whole message in case the message is large and the output of the function is small. Several codes that offer error protection for the output of a machine learning algorithm were given in \cite{MSVD2016,KSVD2017,HSJ2020}. Further, the application of FCCs in archival data storage is detailed in \cite{LBWY2023}. For a function $f$ and a positive integer $t$, a systematic encoding is called an $(f, t)$-FCC if it can protect the function value against at most $t$ errors.

The study of function-correcting codes started with the work in \cite{LBWY2023}, where the authors developed a general theory for FCCs. This research considers systematic codes that focus on reducing redundancy, and the channel considered here is a binary symmetric channel. In this work, FCCs are created for various specific families of functions, such as locally binary functions, Hamming weight functions, and Hamming weight distribution functions. They present some optimal constructions of FCCs for these functions.

Later, the work by Xia et al. in \cite{XLC2016} extends the concept of FCCs to symbol-pair read channels, calling them function-correcting symbol-pair codes (FCSPCs). They also focus on some particular functions and provide constructions for FCSPCs. A recent study by Premlal and Rajan in \cite{PR2024} provides a lower bound on the redundancy of FCCs. Since FCCs are equivalent to error-correcting codes (ECCs) when the function is bijective, this bound is also applicable to systematic ECCs. They show the tightness of this bound for a certain range of parameters. They then focus on function-correcting codes for linear functions, proving that the upper bound proposed by Lenz et al. is tight by providing a construction for these codes for a class of linear functions. 

Recent work by Ge et al. in \cite{GXZZ2025} mainly focuses on two types of functions: Hamming weight functions and Hamming weight distribution functions. They provide some improved bounds on the redundancy of FCCs for these functions and some optimal constructions that achieve the lower bound.
The most recent work by Singh et al. in \cite{SSY2025} extends the work of \cite{XLC2016} for $b$-symbol read channels over finite fields and introduces the idea of irregular $b$-symbol distance codes.

In this work, we generalize the concept of locally binary functions given in \cite{LBWY2023} to a class that we call locally $(\rho, \lambda)$-bounded functions (Definition \ref{rl_func}). For $\lambda=2$, these are the same as $\rho$-locally binary functions. We provide an upper bound on the redundancy of an $(f,t)$-FCC for a subclass of locally $(2t, 4)$-bounded functions, where $t$ is a positive integer, by giving an FCC construction for these functions. Furthermore, we provide an optimality condition for which the upper bound is tight, and the given construction is optimal. 

We also extend this upper bound to the subclass of locally $(2t,\lambda)$-bounded functions whose function balls are contiguous blocks, where the bound depends on the existence of an ECC with certain parameters.


Lastly, we observe that for a fixed positive integer $\rho$, any function between finite sets admitting a Hamming metric can be considered as a locally $(\rho, \lambda)$-bounded function with a suitably chosen $\lambda$, and illustrate this by representing Hamming weight functions and Hamming weight distribution functions as locally $(\rho, \lambda)$-bounded functions. We also provide a simple construction of an FCC for Hamming weight distribution functions using a known error-correcting code.

\textit{Conventions and notation:} A finite field of size $q$ is denoted by $\mathbb{F}_q$. We use the conventions that $\mathbb{N}=\{0,1,2,\dots\}$ and $[n]=\{1, 2, \ldots, n\}$. For any vector $u$, $(u)^t$ represents the $t$-fold repetition of $u$, for example $(011)^2=011011$.  An $(n,M,d)$ code represents a (not necessarily linear) error-correcting code with length $n$, number of codewords $M$, and minimum distance $d$. Finally, we use $N(M,d)$ to denote the smallest length $n$ for which an $(n,M,d)$ code exists. For more details, see \cite{LX2004}.

\section{Preliminaries}
\label{preliminaries}
In this section, we provide some basic concepts and definitions related to function-correcting codes from \cite{LBWY2023}.
\begin{definition}[Function-correcting code (FCC)]
 Consider a function $f: \mathbb{F}_2^k \rightarrow S$, where $S$ depends on the function considered. A systematic encoding $\mathcal{C}: \mathbb{F}_2^k \rightarrow \mathbb{F}_2^{k+r}$ is defined as an \emph{$(f, t)$-FCC} if, for any $u_1, u_2 \in \mathbb{F}_2^k$ such that $f(u_1) \neq f(u_2)$, it holds that: 
$$d(\mathcal{C}(u_1), \mathcal{C}(u_2)) \geq 2t+1,$$
where $d(x, y)$ denotes the Hamming distance between vectors $x$ and $y$.
\end{definition}

Thus, an $(f,t)$-FCC allows the receiver to determine the value $f(u)$ after $t$ bit errors have occurred on the codeword $\mathcal{C}(u)$. If $f: \mathbb{F}_2^k \rightarrow \mathrm{Im}(f)$ is a bijection, then an $(f, t)$-FCC is equivalent to a systematic $(k+r, 2^k, 2t+1)$ error-correcting code. 

\begin{definition}[Optimal redundancy]
The \emph{optimal redundancy} $r_f(k, t)$ is defined as the minimum of $r$ for which there exists an $(f, t)$-FCC with an encoding function $\mathcal{C}: \mathbb{F}_2^k \rightarrow \mathbb{F}_2^{k+r}$.
\end{definition}

\begin{definition}[Distance requirement matrix]
Let $u_1,\ldots, u_{M} \in \mathbb{F}_2^k$. The \emph{distance requirement matrix (DRM)} $\mathcal{D}_f(t, u_1, u_2,\ldots, u_M)$ for an $(f, t)$-FCC is an $M \times M$ matrix with entries
$$
  [\mathcal{D}_f(t, u_1, \ldots, u_M)]_{i, j} = \begin{cases}  \max(2t+1-d(u_i, u_j), 0), & \text{if} \ f(u_i) \neq f(u_j), \\
0 & \text{otherwise},
\end{cases}  
$$
where $i, j \in \{1,2, \ldots, M\}$.
\end{definition}

\begin{example}\label{ex1}
 Consider $\mathbb{F}_2^2 = \{00, 01, 10, 11\}$ and a function $f: \mathbb{F}_2^2 \rightarrow \{0, 1\}$ such that $f(00)= 0, f(01)= f(10)= f(11)=1.$
Then for $t=1$, the distance requirement matrix is

$$\mathcal{D}_f(t, u_1, u_2, u_3, u_4) = \left[\begin{matrix} 0 & 2 & 2 & 1 \\
2 & 0& 0& 0 \\
2 &0 & 0 & 0 \\
1 & 0& 0& 0
\end{matrix}\right].$$

\end{example}

\begin{definition}[Irregular-distance code or $\mathcal{D}$-code]
Let $\mathcal{D} \in \mathbb{N}^{M\times M}$. Then $\mathcal{P}=\{p_1, p_2, \ldots, p_M\}$, where $p_i\in \mathbb{F}_2^r$ for $i \in [M]$, is said to be an \emph{irregular-distance code} or \emph{$\mathcal{D}$-code} if there is an ordering of $\mathcal{P}$ such that $d(p_i, p_j) \geq [\mathcal{D}]_{i,j}$ for all $i, j \in \{1, 2, \ldots, M\}$. Further, $N(\mathcal{D})$ is defined as the smallest integer $r$ such that there exists a $\mathcal{D}$-code of length $r$. If $[\mathcal{D}]_{i, j} = D$ for all $i, j \in \{1,2,\ldots, M\}, i\neq j$, then $N(\mathcal{D})$ is denoted as $N(M, D)$.
\end{definition}

For $\mathcal{D}=\mathcal{D}_f(t, u_1, u_2,\ldots, u_{2^k})$, if we have a $\mathcal{D}$-code $\mathcal{P}=\{p_1, p_2, \ldots, p_{2^k}\}$, then we can use it to construct an $(f, t)$-FCC with the encoding $\mathcal{C}(u_i) = (u_i, p_i)$ for all $i \in \{1, 2, \ldots, 2^k\}$.

\begin{example}
 Consider the same function  $f: \mathbb{F}_2^2 \rightarrow \{0, 1\}$ from Example \ref{ex1}. Then we have   a $\mathcal{D}$-code $\mathcal{P}=\{00, 11, 11, 01\}$ for which the distance structure is 
$$
 \begin{blockarray}{ccccc}
   &  00 & 11 & 11 & 01 \\
  \begin{block}{c[cccc]}
00 & 0 & 2 & 2 & 1 \\
11 & 2 & 0& 0& 1 \\
11 & 2 &0 & 0 & 1 \\
01 & 1 & 1& 1& 0 \\
    \end{block}
\end{blockarray}, $$
 where $\mathcal{D}=\mathcal{D}_f(t, u_1, u_2, u_3, u_4)$, given in Example \ref{ex1}. Since $r=2$ is the smallest length possible for a $\mathcal{D}$-code, we have $N(\mathcal{D}_f(t, u_1, u_2,\ldots, u_M)) = 2$.
Further, the $(f, 1)$-FCC obtained using $\mathcal{P}$ is $\{0000, 0111, 1011, 1101\}.$
\end{example}

\begin{definition}
For a function $f: \mathbb{F}_2^k \rightarrow S$, the \emph{(function) distance} between $f_1, f_2 \in \mathrm{Im}(f)$ is defined as
$$d(f_1, f_2) = \min_{u_1, u_2 \in \mathbb{F}_2^k} \{d(u_1, u_2) |  f(u_1)=f_1, f(u_2) = f_2\}.$$
\end{definition}

\begin{definition}[Function distance matrix]
Consider a function $f: \mathbb{F}_2^k \rightarrow S$ and $E=|\mathrm{Im}(f)|$. Then the $E \times E$ matrix $\mathcal{D}_f(t, f_1, f_2,\ldots, f_E)$ with entries given as
$$
[\mathcal{D}_f(t, f_1, f_2,\ldots, f_E)]_{i, j} = \\ \begin{cases}  \max(2t+1 -d(f_i, f_j), 0), & \text{if} \ i \neq j, \\
0 & \text{otherwise},
\end{cases}
$$
is called a \emph{function distance matrix (FDM)}.
\end{definition}

\begin{example}
\normalfont For the function  $f: \mathbb{F}_2^2 \rightarrow \{0, 1\}$ given in Example \ref{ex1}, we have   
$$\mathcal{D}_f(t=1, f_1=0, f_2=1) = \begin{bmatrix} 0 & 2 \\
2 & 0 
\end{bmatrix}.$$
\end{example}

\begin{theorem}[\kern-0.2em\cite{LBWY2023}]\label{thm1}
For any function $f: \mathbb{F}_2^k \rightarrow S$ and $\{u_1, u_2, \ldots, u_m\}\subseteq \mathbb{F}_2^k$,
$$r_f(k, t) \geq N(\mathcal{D}_f(t, u_1, u_2, \ldots, u_m)),$$
and for $|\mathrm{Im}(f)|\geq 2$, $r_f (k, t) \geq 2t$.
\end{theorem}

\begin{theorem}[\kern-0.2em\cite{LBWY2023}]\label{thm2}
For any function $f: \mathbb{F}_2^k \rightarrow S$,
$$r_f(k, t) \leq N(\mathcal{D}_f(t, f_1, f_2, \ldots, f_E)),$$
where $E=|\mathrm{Im}(f)|$ and $\mathcal{D}_f(t, f_1, f_2, \ldots, f_E)$ is an FDM.
\end{theorem}

\begin{corollary}[\kern-0.2em\cite{LBWY2023}]\label{col1}
If there exists a set of representative   information vectors $u_1,  \ldots, u_{E}$ with $\{f(u_1),  \ldots,$  $f(u_E)\} = \mathrm{Im}(f)$ and $\mathcal{D}_f(t, u_1, \ldots, u_E)=\mathcal{D}_f(t, f_1, \ldots, f_E)$, then
$$r_f(k, t) = N(\mathcal{D}_f(t, f_1, f_2, \ldots, f_E)).$$
\end{corollary}

\begin{theorem}[Plotkin bound \cite{P1960}]\label{thm4}
 Let $A(n, d)$ be the maximal number of possible codewords in a binary code of length $n$ and minimum distance $d$.  If $d$ is even and $2d>n$, then 
$$A(n, d)\leq 2 \left \lfloor \frac{d}{2d-n} \right\rfloor.$$
If $d$ is odd and $2d+1>n$, then
$$A(n, d)\leq 2\left\lfloor {\frac {d+1}{2d+1-n}}\right\rfloor.$$
\end{theorem}
The following bound is a generalization of the Plotkin bound on codes with irregular distance requirements.
\begin{theorem}[\kern-0.2em\cite{LBWY2023}]\label{thm3} 
For any distance matrix $\mathcal{D} \in \mathbb{N}^{M \times M}$,
$$N(\mathcal{D}) \geq \begin{cases}
\frac{4}{M^2} \sum_{i,j,i<j} [D]_{i, j} & \text{if $M$ even,} \\
\frac{4}{M^2-1} \sum_{i, j, i<j} [D]_{i, j} & \text{if $M$ odd}.
\end{cases}$$
\end{theorem}

We use some combinatorial notions in this paper; see \cite{Cameron1994} for details.  
\begin{definition}[Totally ordered set]
A set $P$ with a binary relation $\prec$ is a \emph{totally ordered set} if it satisfies: 
(1) antisymmetry: if $a \prec b$ and $b \prec a$, then $a=b$; 
(2) transitivity: if $a \prec b$ and $b \prec c$, then $a \prec c$; 
(3) totality: for any $a,b \in P$, either $a \prec b$, $b \prec a$, or $a=b$.
\end{definition}
\begin{definition}[Contiguous block]  
A subset $I$ of a totally ordered set $(P,\prec)$ is called a \emph{contiguous block}
(or interval) if for all $a,b\in I$ with $a\prec b$, every $c\in P$
satisfying $a\prec c\prec b$ also belongs to $I$.
\end{definition}


\section{MAIN RESULTS}
\label{main_results}
In this section, we first define locally $(\rho, \lambda)$-bounded functions. Then, we present some results on function-correcting codes for these functions. We also demonstrate that for a fixed $\rho$, any function $f:\mathbb{F}_2^k \rightarrow S,\ |S|<\infty$, will be a locally $(\rho, \lambda)$-bounded function for a suitably chosen value of $\lambda$.

\begin{definition}[Function ball \cite{LBWY2023}]\label{FB}
The \emph{function ball} of a function $f:\mathbb{F}_2^k \rightarrow S$ with radius $\rho$ around $u \in \mathbb{F}_2^k$
is defined by
$$B_f (u, \rho) = \{f(u') | u' \in \mathbb{F}_2^k \ \text{and} \ d(u, u') \leq \rho\}.$$
\end{definition}

\begin{definition}[Locally bounded function] \label{rl_func}
A function $f:\mathbb{F}_2^k \rightarrow S$ is said to be a \emph{locally $(\rho,\lambda)$-bounded} 
function if for all $u \in \mathbb{F}_2^k$,
$$|B_f (u, \rho)| \leq \lambda.$$
\end{definition}


The following lemma will be used in an FCC construction later. Therefore, throughout the remainder of this paper, we assume that all locally $(\rho,\lambda)$-bounded functions under consideration also satisfy the assumption in Lemma~\ref{col}, namely that each neighborhood $B_f(u,\rho)$ forms a contiguous block under a fixed total order on $\mathrm{Im}(f)$. We refer to this assumption as the \emph{contiguity condition}, and all subsequent constructions that make use of Lemma~\ref{col} are presented under this assumption.


\begin{lemma}\label{col}
Let $f:\mathbb{F}_2^k \to S$ be a locally $(\rho,\lambda)$-bounded function.
Assume that there exists a total order $\prec$ on $\operatorname{Im}(f)$ such that for every $u \in \mathbb{F}_2^k$, the set $B_f(u,\rho)$ forms a contiguous block of consecutive elements with respect to $\prec$. Then there exists a mapping
$
\mathrm{Col}_f : \mathbb{F}_2^k \to [\lambda]
$
such that for any $u,v \in \mathbb{F}_2^k$ with $d(u,v) \le \rho$ and $f(u)\neq f(v)$, we have $\mathrm{Col}_f(u)\neq \mathrm{Col}_f(v)$.
\end{lemma}
\begin{proof}
Let $\operatorname{Im}(f) = \{ y_0 \prec y_1 \prec \cdots \prec y_{E-1}\}$ be the image of $f$ ordered by $\prec$. Define a coloring $\gamma:\operatorname{Im}(f)\to[\lambda]$ by $ \gamma(y_j) \;=\; 1 + (j \bmod \lambda).$
By construction, the colors $\{1,2,\dots,\lambda\}$ are assigned cyclically along the order $\prec$. In particular, if $I = \{y_a,y_{a+1},\dots,y_{a+m-1}\}$ is any contiguous block with $m \le \lambda$, then $\gamma$ is injective on $I$, since the residues $a,a+1,\dots,a+m-1$ modulo $\lambda$ are all distinct.
Now define $\mathrm{Col}_f : \mathbb{F}_2^k \to [\lambda]$ by $\mathrm{Col}_f(u) \;=\; \gamma(f(u)).$ 
We claim that $\mathrm{Col}_f$ has the desired separation property. Suppose $u,v \in \mathbb{F}_2^k$ satisfy $d(u,v)\le \rho$ and $f(u)\neq f(v)$. By definition, $f(v)\in B_f(u,\rho)$, so both $f(u)$ and $f(v)$ belong to the same block $B_f(u,\rho)$ of size at most $\lambda$. Since $\gamma$ is injective on every such block, it follows that $\gamma(f(u)) \neq \gamma(f(v))$. Consequently, $\mathrm{Col}_f(u)\neq \mathrm{Col}_f(v)$, as required.
\end{proof}

In the following example, we present a function which satisfies contiguity condition. 

\begin{example}
Consider the following lexicographic rearrangement function on $\mathbb{F}_2^k$:
\[
f : \mathbb{F}_2^k \to \mathbb{F}_2^k, \qquad 
f(u) = 0^{\,k - \mathrm{wt}(u)} 1^{\,\mathrm{wt}(u)},
\]
where $\mathrm{wt}(u)$ denotes the Hamming weight of $u$. 
Equivalently, $f(u)$ is obtained by rearranging the coordinates of $u$ in non-decreasing order. 
For example, we have
$
f(010100) = 000011, f(111001) = 001111.
$

\noindent \textit{Contiguity condition:}
For a fixed $u \in \mathbb{F}_2^k$ and radius $\rho$, the function ball is
$
B_f(u,\rho) = \{ f(v) : d(u,v) \le \rho \}.
$ 
For any $v\in \mathbb{F}_2^k$ with $d(u,v)\le \rho$, we have
\[
\max\{0, \mathrm{wt}(u) - \rho\} \le \mathrm{wt}(v) \le 
\min\{\mathrm{wt}(u) + \rho, k\}.
\]
That means $\mathrm{wt}(v) \in W_{u,\rho}$, where 
$
W_{u,\rho} = \big\{ \max\{0, \mathrm{wt}(u) - \rho\}, \ldots,
           \min\{\mathrm{wt}(u) + \rho, k\} \big\} \in \mathbb{N}.
$
Conversely,  for any element $a\in W_{u,\rho}$, we have a vector $v$ in $\mathbb{F}_2^k$ with $\mathrm{wt}(v)=a$ and $d(u,v) \le \rho$ which can be obtain by changing $a$ number of entries in vector $u$. Thus, we have 
\[
B_f(u,\rho) = \{\, 0^{\,k-j} 1^{\,j} \ : \ j \in W_{u,\rho} \,\}.
\]
Since $W_{u,\rho}$ contains consecutive integers, $B_f(u,\rho)$ forms a contiguous block under the lexicographic order.

\end{example}




\subsection{An FCC for locally $(2t, 4)$-bounded functions}

In this section, we construct an FCC for a locally $(2t, \lambda)$-bounded function $f$ with $\lambda = 4$. Using this construction, we obtain an upper bound on the optimal redundancy of any $(f, t)$-FCC.
The following lemma will be generalized for locally $(2t, \lambda)$-bounded functions for any $\lambda$ in the next subsection, with a proof following the same method. To make it easier to understand, we begin with the case $\lambda = 4$.

\begin{lemma} \label{lem_bound4}
Let $t$ be a positive integer. For any locally $(2t, 4)$-bounded function $f$, the optimal redundancy of an $(f, t)$-FCC is bounded from above as follows.
\begin{equation}\label{bound4}
r_f(k, t) \leq 3t.
\end{equation}
\end{lemma}

\begin{proof}
Let $f$ be a locally $(2t, 4)$-bounded function and $u \in \mathbb{F}_2^k $ be an information symbol vector. By Lemma \ref{col}, there exists a mapping $\mathrm{Col}_f:\mathbb{F}_2^k \rightarrow [4]$ for function $f$ such that for any $u,v \in \mathbb{F}_2^k $, $\mathrm{Col}_f (u) \neq \mathrm{Col}_f(v)$ if $f(u)\ne f(v)$ and $d(u,v)\leq 2t$.
Define an encoding function $\mathrm{Enc}: \mathbb{F}_2^k \rightarrow \mathbb{F}_2^{k+3t}$ as 
\begin{align*}
\mathrm{Enc}(u)&=(u, u_p), \ \text{where} \ u_p=(u'_p)^t \ \text {and} \\
 u'_p &= \begin{cases}
000 & \text{if} \ \mathrm{Col}_f(u)=1, \\
110 & \text{if} \ \mathrm{Col}_f(u)=2, \\
101 & \text{if} \ \mathrm{Col}_f(u)=3, \\
011 & \text{if} \ \mathrm{Col}_f(u)=4. \\
\end{cases}
\end {align*}
 Now we prove that the encoding function defined above is an $(f, t)$-FCC with redundancy $r=3t$. Let $u,v \in \mathbb{F}_2^k$ be such that $f(u) \neq f(v)$. We have 
\begin{equation} \label{mindist}
d(\mathrm{Enc}(u), \mathrm{Enc}(v)) = d(u, v) + d(u_p, v_p).
\end{equation}
There are the following two possible cases with vectors $u$ and $v$.

\textit{Case 1}: If $d(u,v) \geq 2t+1$, then, by \eqref{mindist}, we have
$d(\mathrm{Enc}(u), \mathrm{Enc}(v)) = d(u, v) + d(u_p, v_p)\geq 2t+1.$

\textit{Case 2}: If $d(u,v)\leq 2t$ then $f(v) \in B_f(u, 2t)$, and by the definition of the function $\mathrm{Col}_f:  \mathbb{F}_2^k \rightarrow [4]$, we have $\mathrm{Col}_f(u) \neq \mathrm{Col}_f(v)$. Therefore, $d(u'_p, v'_p) =2$ and $d(u_p, v_p) =t \cdot d(u'_p, v'_p) =2t$. Since $u \neq v$, we have $d(u,v) \geq 1$ and 
$d(\mathrm{Enc}(u), \mathrm{Enc}(v)) = d(u, v) + d(u_p, v_p)\geq 2t+1.$
\end{proof}

\begin{theorem}[Optimality]\label{opti}
For a locally $(2t, 4)$-bounded function $f$ with $|\mathrm{Im}(f)| \geq 3$, if there exists $u_1, u_2, u_3 \in \mathbb{F}_2^k$ with $f(u_i) \neq f(u_j)$ for $i,j\in[3], i\ne j$, such that 
$d(u_1, u_2) =1, d(u_3, u_1)=1 \ \text{and} \ d(u_3, u_2)=2,$
then $r_f(k, t)=3t$ is optimal.
\end{theorem}

\begin{proof}
For $u_1, u_2, u_3 \in \mathbb{F}_2^k$, we have the distance requirement matrix
$$\mathcal{D}_f(t, u_1, u_2, u_3) = \left[\begin{matrix} 0 & 2t & 2t \\
2t & 0& 2t-1 \\
2t &2t-1 & 0  
\end{matrix}\right].$$
From the generalized Plotkin bound given in Theorem \ref{thm3}, we have 
\begin{align*}
N(\mathcal{D}_f(t, u_1, u_2, u_3)) &\geq \frac{4}{3^2-1} (D_{1,2}  + D_{1,3} + D_{2,3}) \\
& \geq \frac{1}{2} (6t-1) = 3t-\frac{1}{2}.
\end{align*}
Since $N(\mathcal{D}_f(t, u_1, u_2, u_3))$ is an integer, by Theorem \ref{thm1} we have 
$r_f(k,t) \geq N(\mathcal{D}_f(t, u_1, u_2, u_3)) \geq 3t.$
Therefore, $r_f(k, t)=3t$.
\end{proof}

\subsection{An upper bound for general locally bounded functions}
The upper bound on the optimal redundancy given in \eqref{bound4} for any locally $(2t, 4)$-bounded function can be generalized for any locally $(2t, \lambda)$-bounded function as follows.

\begin{theorem}\label{lem_bound_lambda}
Let $t$ be a positive integer. For any locally $(2t, \lambda)$-bounded function $f$, the optimal redundancy of an $(f, t)$-FCC is bounded from above as follows.
\begin{equation}\label{bound_lambda}
r_f(k, t) \leq N (\lambda, 2t),
\end{equation}
where $N(\lambda, 2t)$ is the minimum length of a binary error-correcting code with $\lambda$ codewords and minimum distance $2t$. 
\end{theorem}

\begin{proof}
The proof follows in a similar manner as the proof of Lemma \ref{lem_bound4}. Let $f$ be a locally $(2t, \lambda)$-bounded function. 
By Lemma \ref{col}, there exists a mapping $\mathrm{Col}_f:\mathbb{F}_2^k \rightarrow [\lambda]$ for  a function $f$ such that for any $u,v \in \mathbb{F}_2^k $,  $\mathrm{Col}_f (u) \neq \mathrm{Col}_f(v)$ if $f(u)\ne f(v)$ and $d(u,v)\leq 2t$. Let $\mathcal{C}$ be a binary error-correcting code with $\lambda$ codewords, minimum distance $2t$ and length $N(\lambda, 2t)$, and let the codewords of $\mathcal{C}$ be denoted by $C_1, C_2, \ldots, C_{\lambda}$.
Define an encoding function $\mathrm{Enc}: \mathbb{F}_2^k \rightarrow  \mathbb{F}_2^{k+N(\lambda, 2t)}$ as 
$$\mathrm{Enc}(u)=(u, u_p), \ \text{where}  \ u_p =C_{ \mathrm{Col}_f(u)}.$$
 Now we prove that the encoding function defined above is an $(f, t)$-FCC with redundancy $r=N(\lambda, 2t)$. Let $u,v \in \mathbb{F}_2^k$ be such that $f(u) \neq f(v)$. 
Then we have the following two possible cases with vectors $u$ and $v$.

\textit{Case 1}: If $d(u,v) \geq 2t+1$, then we have
$$d(\mathrm{Enc}(u), \mathrm{Enc}(v)) = d(u, v) + d(u_p, v_p)\geq 2t+1.$$

\textit{Case 2}: If $d(u,v) \leq 2t$ then $f(v) \in B_f(u, 2t)$, and by the definition of  $\mathrm{Col}_f$, we have $\mathrm{Col}_f(u) \neq \mathrm{Col}_f(v)$. Therefore, $d(u_p, v_p)=d(C_{\mathrm{Col}_f(u)}, C_{\mathrm{Col}_f(v)}) \geq 2t$ as the minimum distance of the code $\mathcal{C}$ is $2t$. Since $u \neq v$, we have $d(u,v) \geq 1$ and 
$d(\mathrm{Enc}(u), \mathrm{Enc}(v)) = d(u, v) + d(u_p, v_p)\geq 2t+1.$
\end{proof}

\begin{remark}\label{rem1}
      Since the maximum of $\lambda$ for a function $f$ is $E$, where $E=|\mathrm{Im}(f)|$, we have $r_f(k, t) \leq N(\lambda, 2t) \leq N(E, 2t)$. For using an ECC on function values, we would need at least $N(E, 2t+1)$ parities to handle $t$ errors. This also validates the fact that the use of FCCs instead of ECCs on function values is beneficial for any arbitrary function, as depicted in Table I in \cite{LBWY2023} for some particular functions.
\end{remark}

We believe that the following lemma has likely been proven somewhere in the literature. However, since we were unable to locate it, we will provide a brief proof.
\begin{lemma}\label{N_4}
Let $N(\lambda, 2t)$ be the minimum length of a binary error-correcting code with $\lambda$ codewords and minimum distance $2t$. Then $N(4, 2t)=3t$.
\end{lemma}

\begin{proof}
We can easily construct a binary code with length $n=3t$, $M=4$ codewords, and minimum distance $d=2t$  as follows.
\begin{align*}
c_1&= \underbrace{00\ldots0}_{t \ \text{times}}\underbrace{00\ldots0}_{t \ \text{times}}\underbrace{00\ldots0}_{t \ \text{times}},\, 
c_2= \underbrace{11\ldots1}_{t \ \text{times}}\underbrace{11\ldots1}_{t \ \text{times}}\underbrace{00\ldots0}_{t \ \text{times}}, \\
c_3&= \underbrace{11\ldots1}_{t \ \text{times}}\underbrace{00\ldots0}_{t \ \text{times}}\underbrace{11\ldots1}_{t \ \text{times}}, \,
c_4 =\underbrace{00\ldots0}_{t \ \text{times}}\underbrace{11\ldots1}_{t \ \text{times}}\underbrace{11\ldots1}_{t \ \text{times}}.
\end{align*}
Clearly, $N(4, 2t)\leq 3t$. Now we have $2d=4t>3t\geq N(4, 2t)$ and $d$ is even. Using the Plotkin bound given in Theorem~\ref{thm4} for these parameters, we have 
$$4 \leq 2 \left \lfloor \frac{2t}{4t-n} \right \rfloor \leq 2 \left ( \frac{2t}{4t-n} \right ),$$
which implies that $N(4, 2t) \geq n \geq 3t$.  Therefore, $N(4, 2t)=3t$.
\end{proof}
We note that Lemma \ref{lem_bound4} can also be obtained using Theorem \ref{lem_bound_lambda} and Lemma \ref{N_4}.

\subsection{Arbitrary functions as $(\rho, \lambda)$-bounded functions}
Any function can be considered as a locally $(\rho, \lambda)$-bounded function for some values of $\rho$ and $\lambda$. Suppose that we have a function $f:\mathbb{F}_2^k \rightarrow S,\ |S|<\infty$, and want to construct an $(f, t)$-FCC for it. Then we may select the minimal $\lambda$ for which this function is a locally $(2t, \lambda)$-bounded function, which is
\begin{equation}\label{find_l} 
\lambda= \max_{u\in \mathbb{F}_2^k} |B_f(u, 2t)|.
\end{equation}
 Furthermore, if we have a binary error-correcting code with length $N$, $\lambda$ codewords and minimum distance $2t$,  we can construct an $(f, t)$-FCC with redundancy $r=N$ by using the  construction given in the proof of Theorem \ref{lem_bound_lambda}.

Next, we analyze some existing classes of functions on $\mathbb{F}_2^k$ as locally $(2t, \lambda)$-bounded functions.
\begin{itemize}
\item \textbf{Hamming weight function:} This function is defined as $f:\mathbb{F}_2^k\rightarrow \mathbb{N},\ f(u)=wt(u)$, where $wt(u)$ denotes the Hamming weight of the vector $u$. 

\item \textbf{Hamming weight distribution function:} For a given positive integer $T$ called a \emph{threshold}, this function is defined as $f:\mathbb{F}_2^k\rightarrow \mathbb{N},\ f(u)=\Delta_T(u)=\left \lfloor \frac{wt(u)}{T} \right \rfloor$. 
\end{itemize}
Note that the Hamming weight function is a Hamming weight distribution function with threshold $T=1$.

The following lemma gives a suitable value of $\lambda$ for the Hamming weight distribution function with threshold $T$.

\begin{theorem}\label{HWDF}
For a positive integer $t$, a Hamming weight distribution function $\Delta_T$ is a locally $\left( 2t, \left \lfloor \frac{4t}{T} \right \rfloor +2 \right)$-bounded function.
\end{theorem} 

\begin{proof}
We prove this theorem by finding a suitable $\lambda$ as given in \eqref{find_l}. For a Hamming weight distribution function $f=\Delta_T$ with threshold $T$, we claim that 
$$ \max_{u\in \mathbb{F}_2^k} |B_f(u, 2t)| \leq  \left \lfloor \frac{4t}{T} \right \rfloor +2.$$
On the contrary, assume that $ |B_f(u, 2t)| >  \left \lfloor \frac{4t}{T} \right \rfloor +2$ for some $u \in \mathbb{F}_2^k$. Then there exist $v_1, v_2 \in  \mathbb{F}_2^k$ with $d(u, v_1)\leq 2t$ and $d(u, v_2)\leq 2t$ such that $f(v_1)= \max(B_f(u, 2t))$ and $f(v_2)= \min(B_f(u, 2t))$. Clearly, 
$$f(v_1)-f(v_2) \geq \left \lfloor \frac{4t}{T} \right \rfloor +2.$$
Further, we have $\left \lfloor \frac{wt(v_1) -wt(v_2)}{T} \right \rfloor   \geq \left \lfloor \frac{wt(v_1)}{T} \right \rfloor - \left \lfloor \frac{wt(v_2 )}{T} \right \rfloor -1 =f(v_1)-f(v_2)  -1  \geq \left \lfloor \frac{4t}{T} \right \rfloor +1.$
Since $d(v_1, v_2) \leq d(u, v_1) +d(u, v_2) \leq 4t$, we have $wt(v_1) - wt(v_2) \leq d(v_1, v_2) \leq 4t$. Therefore,
\begin{align*}
 \left \lfloor \frac{4t}{T} \right \rfloor  \geq \left \lfloor \frac{wt(v_1) -wt(v_2)}{T} \right \rfloor   & \geq \left \lfloor \frac{4t}{T} \right \rfloor +1, 
\end{align*}
which is a contradiction. Hence $ |B_f(u, 2t)| \leq  \left \lfloor \frac{4t}{T} \right \rfloor +2$ for all $u \in \mathbb{F}_2^k$. Therefore, the function $f$ is a locally $\left( 2t, \left \lfloor \frac{4t}{T} \right \rfloor +2 \right)$-bounded function.
\end{proof}

From Theorem \ref{HWDF}, we directly get the following corollaries.
\begin{corollary}\label{col2}
For a positive integer $t$, the Hamming weight function is a locally $(2t, 4t+2)$-bounded function. 
\end{corollary}

\begin{corollary}\label{col3}
Let $t$ be a positive integer. For the Hamming weight distribution function $\Delta_T$,
\begin{itemize}
\item if $T> 4t$, then $\Delta_T$ is a locally $(2t, 2)$-bounded function or, equivalently, a $2t$-locally binary function.
\item if $4t \geq T > 2t$, then $\Delta_T$ is a locally $(2t, 3)$-bounded function.
\item in general, if $\frac{4t}{m-1} \geq T > \frac{4t}{m}$ for $m \in \{2, 3, \ldots, 4t\}$, then $\Delta_T$ is a locally $\left( 2t, m+1\right)$-bounded function.
\end{itemize}
    
\end{corollary}

An optimal construction of an $(f, t)$-FCC has been given for any locally $(2t, 2)$-bounded function $f$ in \cite{LBWY2023} with redundancy $r_{f}(k, t)=2t$. Furthermore, another optimal construction of a $(\Delta_T, t)$-FCC called Construction 2 was proposed in \cite{LBWY2023} for $4t \geq T > 2t$. Here, we present another simple optimal construction for a $(\Delta_T, t)$-FCC for $4t \geq T > 2t$.

Define an encoding function $\mathrm{Enc}: \mathbb{F}_2^k \rightarrow  \mathbb{F}_2^{k+2t}$ as 
$$\mathrm{Enc}(u)=(u,u_p), \ \text{where} \ u_p=(\Delta_T(u) \bmod{2})^{2t}.$$

 Now we prove that the encoding function defined above is a $(\Delta_T, t)$-FCC with redundancy $r=2t$. Let $u,v \in \mathbb{F}_2^k$ be such that $\Delta_T (u) \neq \Delta_T(v)$. Then we have the following cases.
 
\textit{Case 1}: If $\Delta_T(u)$ and $\Delta_T(v)$ have different parity, then by the definition of the encoding function, $d(u_p, v_p) =2t$. Since $u \neq v$, we have $d(u,v) \geq 1$ and 
$d(\mathrm{Enc}(u), \mathrm{Enc}(v)) = d(u, v) + d(u_p, v_p)\geq 2t+1.$

\textit{Case 2}:  If $\Delta_T(u)$ and $\Delta_T(v)$ have the same parity, then WLOG assuming $\Delta_T (u)>\Delta_T(v)$, we have $$\Delta_T (u)-\Delta_T(v)  \geq 2.$$
Further, we have 
\begin{align*}
 \frac{wt(u)-wt(v)}{T} & \geq  \left \lfloor \frac{wt(u)-wt(v)}{T} \right \rfloor  \\
&\geq  \left \lfloor \frac{wt(u)}{T} \right \rfloor - \left \lfloor \frac{wt(v )}{T} \right \rfloor - 1\\
 &= \Delta_T (u)-\Delta_T(v) -1 \geq 1,
\end{align*}
which implies that $d(u,v) \geq wt(u)-wt(v)  \geq T \geq 2t+1.$
Therefore, we have $d(\mathrm{Enc}(u), \mathrm{Enc}(v)) = d(u, v) + d(u_p, v_p)\geq 2t+1.$

This construction can be generalized for any Hamming weight distribution function $\Delta_T$ with $\frac{4t}{m-1} \geq T > \frac{4t}{m}$, and this will provide an upper bound on the redundancy of a $(\Delta_T, t)$-FCC as given in the following theorem.

\begin{theorem}\label{thm_bound_H}
For the Hamming weight distribution function $\Delta_T$, where $\frac{4t}{m-1} \geq T >\frac{4t}{m}$ for some integer $m \in \{2,3, \ldots, 4t\}$, the optimal redundancy of an FCC is bounded from above as follows.
\begin{equation}\label{bound_H}
r_{\Delta_T}(k, t) \leq N \left(\left \lceil \frac{m}{2} \right \rceil +1, 2t \right),
\end{equation}
where $N(M, d)$ represents the minimum length of a binary error-correcting code with $M$ codewords and minimum distance $d$. Furthermore,
$$r_{\Delta_T}(k, t) \leq 3t \quad \text{if} \ t \geq T > \frac{2t}{3}.$$
\end{theorem}

\begin{proof}
To prove this theorem we propose the following construction for a $(\Delta_T, t)$-FCC, where $\frac{4t}{m-1} \geq T >\frac{4t}{m}$ and $m \in \{2,3, \ldots, 4t\}$.

\noindent \textbf{Construction:} Let $a=\left\lceil \frac{m}{2} \right \rceil +1$ and $\mathcal{C}$ be a binary error-correcting code with $a$ codewords, minimum distance $2t$, and length $N(a, 2t)$. Let the codewords of $\mathcal{C}$ be denoted by $C_0, C_1, \ldots, C_{a-1}$. Define an encoding function $\mathrm{Enc}: \mathbb{F}_2^k \rightarrow  \mathbb{F}_2^{k+N(a, 2t)}$ as 
$$\mathrm{Enc}(u)=(u, u_p), \ \text{where} \  u_p = C_{\Delta_T(u) \bmod{a}}.$$
 Now we prove that the encoding function defined above is a $(\Delta_T, t)$-FCC with redundancy $r=N(a, 2t)$. Let $u,v \in \mathbb{F}_2^k$ such that $\Delta_T (u) \neq \Delta_T(v)$. Then we have the following cases.

\textit{Case 1}: If $0<|\Delta_T(u)-\Delta_T(v)| \leq a-1$, then
$C_{\Delta_T(u) \bmod{a}} \neq C_{\Delta_T(v) \bmod{a}}$. Since the minimum distance of the code $\mathcal{C}$ is $2t$, we have $d(u_p, v_p) \geq 2t$. Since $u \neq v$, we have $d(u,v) \geq 1$ and 
$d(\mathrm{Enc}(u), \mathrm{Enc}(v)) = d(u, v) + d(u_p, v_p)\geq 2t+1.$

\textit{Case 2}:  If $|\Delta_T(u) - \Delta_T(v)| > a-1$, then WLOG assuming $\Delta_T (u)>\Delta_T(v)$, we have $\Delta_T(u)-\Delta_T(v)  \geq a$. Further,
\begin{align*}
 \frac{wt(u)-wt(v)}{T} 
& \geq \left \lfloor \frac{wt(u)}{T} \right \rfloor - \left \lfloor \frac{wt(v )}{T} \right \rfloor-1 \\
 &=\Delta_T(u)-\Delta_T(v) -1 \geq a-1.
\end{align*}
Since $d(u,v) \geq wt(u)-wt(v)$, we have 
$$d(u, v) \geq (a-1) T  = \left\lceil \frac{m}{2} \right \rceil   T > \left\lceil \frac{m}{2} \right \rceil  \frac{4t}{m} \geq 2t. $$
Therefore, we have $d(u,v) \geq 2t+1$ and
$d(\mathrm{Enc}(u), \mathrm{Enc}(v)) = d(u, v) + d(u_p, v_p)\geq 2t+1.$

Since we know $N(4, 2t)=3t$, and for $m=5$ and $m=6$, we have $a=\left\lceil \frac{m}{2} \right \rceil +1=4$. Therefore, 
$r_{\Delta_T}(k, t) \leq 3t \quad \text{if} \ t \geq T > \frac{2t}{3}.$
\end{proof}

As described in Corollary \ref{col3}, a Hamming weight distribution function $\Delta_T$, where $\frac{4t}{m-1} \geq T > \frac{4t}{m}$ and $m \in \{2, 3, \ldots, 4t\}$, is a locally $(2t, m+1)$-bounded function. 
Also, the function ball for $u\in\mathbb{F}_2^k$ and radius $\rho$ can be expressed as
\[
B_f(u,\rho)=\Big\{\,\Big\lfloor \frac{w}{T}\Big\rfloor : w\in W_{u,\rho}\,\Big\},
\]
where $
W_{u,\rho} = \big\{ \max\{0, \mathrm{wt}(u) - \rho\}, \ldots,
           \min\{\mathrm{wt}(u) + \rho, k\} \big\}.
$ Since $W_{u,\rho}$ consists of consecutive integers, its image under the floor operation is also a consecutive set of integers. Hence,
\[
B_f(u,\rho)=\{\,j\in\mathbb{N}: j_{\min}\le j\le j_{\max}\,\},
\]
where $j_{\min}=\big\lfloor \tfrac{\max\{0, \, \mathrm{wt}(u)-\rho\}}{T}\big\rfloor$ and 
$j_{\max}=\big\lfloor \tfrac{\min\{\mathrm{wt}(u)+\rho, \,k\}}{T}\big\rfloor$.  
Therefore, $B_f(u,\rho)$ forms a contiguous block under the natural total order on $\mathbb{N}$, and the Hamming weight distribution function satisfies the contiguity condition. So the bound in \eqref{bound_lambda} applies, giving $$r_{\Delta_T}(k, t) \le N(m + 1, 2t).$$
 It is, however, evident that Theorem~\ref{thm_bound_H} provides a tighter bound for $\Delta_T$, and it relies on the existence of an error-correcting code with specific parameters. For the Hamming weight distribution function, another tight upper bound is given by a construction in \cite{GXZZ2025}.

\section{Conclusion}

\label{conclusion}
In this work, locally $(\rho, \lambda)$-bounded functions were defined and related function-correcting codes were considered along with redundancy bounds. 
As the results rely on Lemma \ref{col}, they apply to the subclass of locally bounded functions where each function ball forms a contiguous block in a fixed total order. In particular, this subclass includes important functions such as the Hamming weight distribution function. In this study, all codes were binary, as the domain of the functions was considered to be $\mathbb{F}_2^k$. The work on general locally $(\rho, \lambda)$-bounded functions over $\mathbb{F}_q^k$ is currently in progress.

\section*{Acknowledgment}
This work was supported by a joint project grant to Aalto University and Chalmers University of Technology (PIs A. Graell i Amat and C. Hollanti) from the Wallenberg AI, Autonomous Systems and Software Program, and additionally by the Science and Engineering Research Board (SERB) of the Department of Science and Technology (DST), Government of India, through the J.C. Bose National Fellowship to Prof. B. Sundar Rajan.

\bibliographystyle{ieeetr}
\bibliography{FCC}

\end{document}